\documentclass[10pt,twocolumn,twoside]{IEEEtran}
\usepackage{resizegather}

\usepackage{textcomp}
\usepackage{graphicx}
\usepackage{graphics} 
\usepackage{epsfig} 
\usepackage[abs]{overpic}
\usepackage{algorithmic}
\usepackage{epstopdf}
\graphicspath{{figs/}}
\usepackage{setspace}
\let\labelindent\relax
\usepackage{enumitem}
\usepackage{amsmath}
\usepackage{amsfonts}

\usepackage{mathtools}
\usepackage{color}

\usepackage{amsthm} 
\newtheorem{theorem}{\bf{Theorem}}[section]

\newtheorem{lem}[theorem]{Lemma}

\theoremstyle{plain}

\newenvironment{definition}[1][Definition]{\begin{trivlist}
\item[\hskip \labelsep {\bfseries #1}]}{\end{trivlist}}

\newcounter{assump} 
\newcounter{rem} 
\newenvironment{remark}[1][Remark \arabic{rem}]{\refstepcounter{rem} \begin{trivlist} 
\item[\hskip \labelsep {\bfseries #1}]}{\end{trivlist}}

\newcounter{exm} 
\newenvironment{example}[1][Example \arabic{exm}]{\refstepcounter{exm} \begin{trivlist} 
\item[\hskip \labelsep {\bfseries #1}]}{\end{trivlist}}

\newcounter{propno} 
\mathtoolsset{showonlyrefs}

\newcounter{algno} 
\begin{document}
\title{
Absorption in Time-Varying Markov Chains: Graph-Based Conditions
}
\author{Yasin Yaz{\i}c{\i}o\u{g}lu
\thanks{Yasin~Yaz{\i}c{\i}o\u{g}lu is with the Department of Electrical and Computer Engineering at the University of Minnesota, Minneapolis, MN, USA (e-mail: ayasin@umn.edu).} 
}


\maketitle
\begin{abstract}
We investigate absorption, i.e., almost sure convergence to an absorbing state, in time-varying (non-homogeneous) discrete-time Markov chains with finite state space. We consider systems that can switch among a finite set of transition matrices, which we call the modes. Our analysis is focused on two properties: 1) almost sure convergence to an absorbing state under any switching, and 2) almost sure convergence to a desired set of absorbing states via a proper switching policy. We derive necessary and sufficient conditions based on the structures of the transition graphs of modes. More specifically,  we show that a switching policy that ensures almost sure convergence to a desired set of absorbing states from any initial state exists if and only if those absorbing states are reachable from any state on the union of simplified transition graphs. We then show three sufficient conditions for absorption under arbitrary switching. While the first two conditions depend on the acyclicity (weak acyclicity) of the union (intersection) of simplified transition graphs, the third condition is based on the distances of each state to the absorbing states in all the modes. These graph theoretic conditions can verify the stability and stabilizability of absorbing states based only on the feasibility of transitions in each mode.

\end{abstract}
\begin{IEEEkeywords}
Markov processes, switched systems, stochastic systems
\end{IEEEkeywords}


\section{Introduction}

\IEEEPARstart{M}{any} natural and engineered systems  involve stochastic dynamics that can be modeled using the framework of Markov chains. Examples include social, biological, and financial systems; transportation, energy, sensor, and communication networks; and robotics.
For systems operating in dynamic environments, where the transition probabilities among states may change over time, the corresponding models become time-varying (non-homogeneous). Our focus in this paper will be on such time-varying models that can switch among a finite set of transition matrices, which we will call the modes of the system. Many studies have investigated the long-run behavior of time-varying Markov chains with ergodic modes (e.g., \cite{lim2013robustness,saloff2009merging,coppersmith2008conditions}). On the other hand, the long-run behavior of time-varying Markov chains with absorbing states is relatively under-explored. A state is called absorbing if it is impossible
to leave it. Markov chains with absorbing states appear in many areas such as optimization (e.g., \cite{shi2000nested}), game theory (e.g., \cite{Young04}), formal methods and verification (e.g., \cite{belta2017formal}), epidemiology (e.g., \cite{allen2008introduction}), motion planning and navigation (e.g., \cite{kurniawati2011motion}). For such systems, reaching an absorbing state corresponds to reaching a local optima (optimization), a Nash equilibrium (game theory), an accepting state of automaton (formal verification), an all-healthy state (epidemiology), or a desired position (motion planning). For time-invariant systems, quantities such as the probability of reaching a specific absorbing state, expected time to reach an absorbing state, or average time spent at each transient state can be easily computed based on the transition probabilities (e.g., \cite{van2009quasi,grinstead2012introduction}). However, such methods do not translate to time-varying Markov chains in general.





In this paper, we focus on time-varying Markov chains with absorbing states and we aim to address two questions: 1) Will the system almost surely reach an absorbing state from any initial condition under any switching among the modes? 2) Can the system be almost surely driven to a set of desired absorbing states by properly switching among the modes? Such stability and stabilizability questions regarding switching systems have been extensively studied in the control theory literature (e.g., \cite{Lin09stability,Ye1998stability,Hespanha99stability,Goebel09hybrid,Liberzon03switching}). There are also many studies on the mean-square stability of systems under exogenous noise and stochastic switching (e.g., \cite{krstic1998stabilization,costa2006discrete}). These works mainly investigate the algebraic conditions for stability of such switching systems. In this paper, we present a graph-theoretic analysis for time-varying discrete-time Markov chains with finite state space. The main contributions of this paper are as follows:



\begin{itemize}

\item  We show that a switching policy that ensures almost sure convergence to a desired set of absorbing states from any initial state exists if and only if that set of absorbing states is reachable from any state on the union of simplified transition graphs (Theorem \ref{conv2a}). 


\item  We show that almost sure convergence to an absorbing state from any initial condition under any switching is possible only when all the modes have the same set of absorbing states (Lemma \ref{conv-nec}) and we provide three sufficient conditions for such stability (Theorem \ref{conv1}): 1) the intersection of simplified graphs is weakly acyclic and has no sinks other than the absorbing states, or 2) the union of simplified transition graphs is acyclic, or 3) in every mode, each non-absorbing state $a_i$ has a feasible transition to some state $a_j$ whose maximum distance (among all simplified transition graphs) to the set of absorbing states is less than that of $a_i$'s. Each of these sufficient conditions can verify stability in some cases where the other two conditions are not satisfied.
\end{itemize}

The organization of this paper is as follows: Section \ref{prelim} provides some graph theory preliminaries. Section \ref{prob} presents the problem formulation. Section \ref{main} presents our main results. Finally, Section \ref{conc} concludes the paper.



\section{Preliminaries}
\label{prelim}




In this section, we present some graph preliminaries that will be used in our analysis. A directed graph $\mathcal{G}=(V,E)$ consists of a \emph{node set} $V$ and an \emph{edge set} $E \subseteq V \times V$. 
A \emph{path} is a sequence of nodes such that each node is adjacent to the preceding node in the sequence. The length of a path is equal to the number of edges traversed. For any two nodes $v$ and $v'$, the \emph{distance} $d(v,v')$ is the number of edges on a shortest path from $v$ and $v'$. We follow the convention that $d(v,v)=0$ and $d(v,v') =\infty$ if $v'$ is not reachable from $v$. Similarly, we define the distance of any node $v$ to a set of nodes $V' \subseteq V$ as the minimum distance between $v$ and the nodes in $V'$, i.e., $
d(v,V') = \min_{v'\in V'} d(v,v')$.

A \emph{sink} is a node with no outgoing edges. A directed graph ${\mathcal{G}=(V,E)}$ is \emph{acyclic} if there is no feasible path that starts and ends at the same node. It is
\emph{weakly acyclic} if there is a feasible path from any node to one of the sink nodes. 

For any two graphs, $\mathcal{G}=(V,E)$ and $\mathcal{G}'=(V',E')$, the union graph and the intersection graph are defined as
\begin{equation}
\label{uni}
\mathcal{G} \cup \mathcal{G}'= (V\cup V', E\cup E'),
\end{equation}
\begin{equation}
\label{inter}
\mathcal{G} \cap \mathcal{G}'= (V\cap V', E\cap E').
\end{equation}

\section{Problem Formulation}
\label{prob}


Consider a time-varying (non-homogeneous) Markov chain over a fine state space ${A=\{a_1,a_2,\hdots,a_n\}}$, where the transition matrix always belongs to a finite set ${\pi=\{P_1,P_2, \hdots, P_k\}}$. We refer to each $P_i \in \pi$ as a \emph{mode} of the system. For each ${i\in \{1,2,\hdots,n\}}$, let $x_i(t)\in [0,1]$ denote the probability of having $a(t)=a_i$. Accordingly, for any arbitrary initial state $a(0)=a_i$, $x(0)$ is obtained by setting $x_i(0)=1$ and $x_j(0)=0$ for all $j\neq i$. Starting with this initial condition,  $x(t)$ evolves under
\begin{equation}
    \label{mc}
  x^{\text{T}}(t+1)=x^{\text{T}}(t)P_{\sigma(t)},
\end{equation}
where $\sigma(t) \in \{1,2,\hdots,k\}$ is the \emph{switching signal} denoting which mode is active at time $t$. When the switching signal is defined as a function of state, we will refer to it as a \emph{switching policy} and, with a slight abuse of notation, we will denote it as $\sigma(a)$. For any $P_{\sigma(t)}\in \pi$, $[P_{\sigma(t)}]_{ij}$, which is the entry in the $i^{th}$ row and $j^{th}$ column of $P_{\sigma(t)}$, denotes the probability of having $a(t+1)=a_j$ given that $a(t)=a_i$. For each $P_i\in \pi$, $A_i^* \subseteq A$ denotes the set of absorbing states in mode $P_i$, i.e.,
\begin{equation}
    A^*_i= \{a_j \in A \mid [P_i]_{jj}=1\}.\end{equation}

We say that $P_i$ is an absorbing mode if it allows for a finite sequence of transitions from every state to some state in $A^*_i$. The set of all absorbing states and and the set of common absorbing states will be denoted as $A^*_\cup$ and $A^*_\cap$, i.e.,
  \begin{equation}
    \label{uabset}
A^*_\cup = A^*_1 \cup \hdots \cup A^*_k, 
\end{equation}
 \begin{equation}
    \label{cabset}
A^*_\cap = A^*_1 \cap \hdots \cap A^*_k. 
\end{equation} 

 

In this paper, we investigate the limiting behavior of $a(t)$ in such time-varying systems with absorbing states. In particular, we aim to address two questions: 
\begin{enumerate}
    \item Would $a(t)$ almost surely converge to an absorbing state in $A^*_\cup$ from any initial state $a(0) \in A$, 
    i.e., 
    \begin{equation}
    \label{asc}
\sum_{a^* \in A^*_\cup}\Pr \left [\lim_ {t \rightarrow \infty}  a(t) =a^* \right ] = 1, 
\end{equation}
under any  switching signal $\sigma(t)$?  
    
\item Given a set of desired absorbing states $A^*_{goal} \subseteq A^*_\cup$, is there a switching policy $\sigma: A \mapsto \{1,2, \hdots,k\}$ that guarantees almost sure convergence to an absorbing state in $A^*_{goal}$, i.e., 
    \begin{equation}
    \label{asc2}
\sum_{a^* \in A^*_{goal}}\Pr \left [\lim_ {t \rightarrow \infty}  a(t) =a^* \right ] = 1, 
\end{equation} from any initial state $a(0) \in A$?

\end{enumerate}

 While the first question is important when investigating the robustness of absorption to changes in the system (e.g., a dynamic environment determining the mode), the second question is important when the goal is to stabilize the system to a specific set of absorbing states via controlled switching. We will answer these questions based on the structures of \emph{simplified transition graphs} associated with $\pi$. 
 

\begin{definition} (Simplified Transition Graph) Given a transition matrix $P$ over the state space $A$, the simplified transition graph $\mathcal{G}=(A,E)$ is a simple directed graph that denotes all the feasible transitions other than the self-loops, i.e.,
\begin{equation}
    \label{eset}
E = \{(a_i,a_j) \mid a_i \neq a_j \in A, \; [P]_{ij} \neq 0 \}. 
\end{equation}
\end{definition}


For any set of transition matrices, $\pi=\{P_1, \hdots, P_k\}$, we will use $\mathcal{G}_1, \hdots, \mathcal{G}_k$ to denote the corresponding simplified transition graphs. Furthermore, we will use $\mathcal{G}_\cup$ and $\mathcal{G}_\cap$ to denote the intersection and union of those graphs, i.e.,
 \begin{equation}
    \label{ugraph}
\mathcal{G}_\cup = \mathcal{G}_1 \cup \hdots \cup \mathcal{G}_k, 
\end{equation}
 \begin{equation}
    \label{cgraph}
\mathcal{G}_\cap = \mathcal{G}_1 \cap \hdots \cap \mathcal{G}_k. 
\end{equation} 
Note that the simplified transition graphs are determined by the feasibility of transitions rather than the exact transition probabilities. Hence, when giving examples in our analysis, we will express the modes as structure matrices, e.g.,
 \begin{equation}
    \label{modex}
P=\begin{bmatrix} \times & \times & \times & 0 \\ 0 & 1 & 0 & 0 \\ 0 & \times & 0 & \times \\0 & 0 & 0 & 1 \end{bmatrix},
\end{equation}
 where $[P]_{ij}=\times$ denotes that the probability of transition from state $a_i$ to $a_j$ is in $(0,1)$ whereas the probabilities that are equal to zero or one are explicitly given in $P$. For all such examples, our discussions and analysis hold for any choice of transition matrices with the specified structures.

\section{Main Results}
\label{main}
In this section, we present our main results regarding the stability and stabilazibility of absorbing states. 

\subsection{Stabilazibility of Absorbing States}
In this subsection, we investigate the stabilizability of absorbing states. We first present a necessary and sufficient condition for the existence of a switching policy that ensures almost sure convergence to a desired set of absorbing states $A^*_{goal} \subseteq A^*_\cup$. In particular, we show that such a switching policy exists if and only if $A^*_{goal}$ is reachable from any state on the union of simplified transition graphs.


\begin{theorem}
\label{conv2a} Let $\pi=\{P_1, \hdots, P_k\}$  be transition matrices over a finite state space $A$ such that there is at least one absorbing state, $A^*_\cup \neq \emptyset$. For any $A^*_{goal} \subseteq A^*_\cup$, there exists a switching policy  $\sigma: A \mapsto \{1,2,\hdots,k\}$ that achieves almost sure convergence to an absorbing state in $A^*_{goal}$ from any initial state if and only if 
\begin{equation}
    \label{conv21}
d_\cup(a,A^*_{goal})<\infty, \forall a \in A, 
\end{equation}
where $d_\cup(a,A^*_{goal})$ denotes the distance of $a$ to $A^*_{goal}$ on the union of simplified transition graphs, $\mathcal{G}_\cup$.
\end{theorem} 
\begin{proof} ($\Rightarrow$:) Suppose that there exists a switching policy ${\sigma: A \mapsto \{1,2,\hdots,k\}}$ that guarantees almost sure convergence to some $a^* \in A^*_{goal}$ from any $a(0)\in A$, but there exists $a' \in A$ such that $d_\cup(a',A^*_{goal})= \infty$, i.e., there is no feasible path from $a'$ to $A^*_{goal}$ on $\mathcal{G}_\cup $. Note that any feasible transition the system can make between two distinct states in any mode is contained as an edge in $\mathcal{G}_\cup$. Hence, if $d_\cup(a',A^*_{goal})= \infty$, then starting from the initial condition $a(0)= a'$ there is no feasible trajectory to $A^*_{goal}$, which contradicts with the existence of a switching policy achieving almost sure convergence to $A^*_{goal}$.

($\Leftarrow$:) Let \eqref{conv21} hold.  Consider a switching policy ${\sigma: A \mapsto \{1,2,\hdots,k\}}$ that maps 1) each state $a_i \notin A^*_{goal}$ to a mode that allows for a transition to a state closer to $A^*_{goal}$ and 2) each state $a_i \in A^*_{goal}$ to any of the modes in $\{1,\hdots,k\}$. More specifically, let 
\begin{equation}
\label{closer}
C(a_i) = \{a_j \in A \mid  d_\cup(a_i,A^*_{goal}) > d_\cup(a_j,A^*_{goal})\}. 
\end{equation} 
Accordingly, the proposed switching policy is 
\begin{equation}
\label{com-pot}
\sigma(a_i) \in \left\{\begin{array}{ll} \{s \mid \exists a_j \in C(a_i) \mbox{ s.t. } [P_s]_{ij}>0 \}, &  \forall a_i \notin A^*_{goal}, \\ \{s \mid a_i \in A_s^*\},&\forall a_i \in A^*_{goal},\end{array}\right.
\end{equation} 
where $\sigma(a_i)$ can be any arbitrary element from those sets. Note that $\{s \mid \exists a_j \in C(a_i) \mbox{ s.t. } [P_s]_{ij}>0 \} \neq \emptyset$ for any $a_i \notin A^*_{goal}$ as otherwise $A^*_{goal}$ is not reachable from $a_i$ on the union graph and \eqref{conv21} would be violated. Also, ${\{s \mid a_i \in A_s^*\} \neq \emptyset}$ for any  $a_i \in A^*_{goal}$ since $A^*_{goal} \subseteq A^*_\cup$.
 Note that \eqref{conv21} also implies that the maximum distance of any state to $A^*_{goal}$ on ${\mathcal{G}_\cup}$ is finite, i.e.,
\begin{equation}
\label{com-pot2}
m= \max_{a \in A} d_\cup(a,A^*) < \infty.
\end{equation} 
Accordingly, there is always a non-zero probability that the system will reach an absorbing state in $A^*_{goal}$ within the next $m$ time steps by constantly moving closer to $A^*_{goal}$ under the proposed switching policy. Repeating this over intervals of $m$-steps, we can show that the probability of this event (reaching $A^*_{goal}$ within the next $m$ time steps) never happening converges to zero as time goes to infinity. 
 \end{proof}

Theorem \ref{conv2a} provides an exact characterization of the stabilizability of any desired set of absorbing states $A^*_{goal} \subseteq A^*_\cup$. One trivial example is when all the modes have the same absorbing states, i.e., $A^*_\cup=A^*_\cap$, and one of the modes has a weakly acyclic simple transition graph. In that case, constantly staying in such a weakly acyclic mode would ensure almost sure convergence to the set of absorbing states. There are also more complicated cases where none of the modes can ensure almost sure convergence to absorbing states from every initial condition whereas a properly designed switching among the modes can achieve that. We provide such an example below.

\begin{example}
Consider a system with two modes:
\begin{equation}
    \label{modex1}
P_1=\begin{bmatrix} 0 & \times & \times & 0 \\ 0 & 1 & 0 & 0 \\ 0 & \times & 0 & \times \\0 & 0 & 0 & 1 \end{bmatrix}, \; P_2=\begin{bmatrix} 0 & \times & \times & 0 \\ 0 & 1 & 0 & 0 \\ 0 & 0 & 0 & 1 \\0 & 0 & \times & \times \end{bmatrix}.
\end{equation}
Accordingly, the simplified transition graphs $\mathcal{G}_1$ and $\mathcal{G}_2$, and their union are as follows:
\vskip2ex
\begin{center} \includegraphics[trim =0mm 0mm 0mm 0mm, clip,scale=0.4]{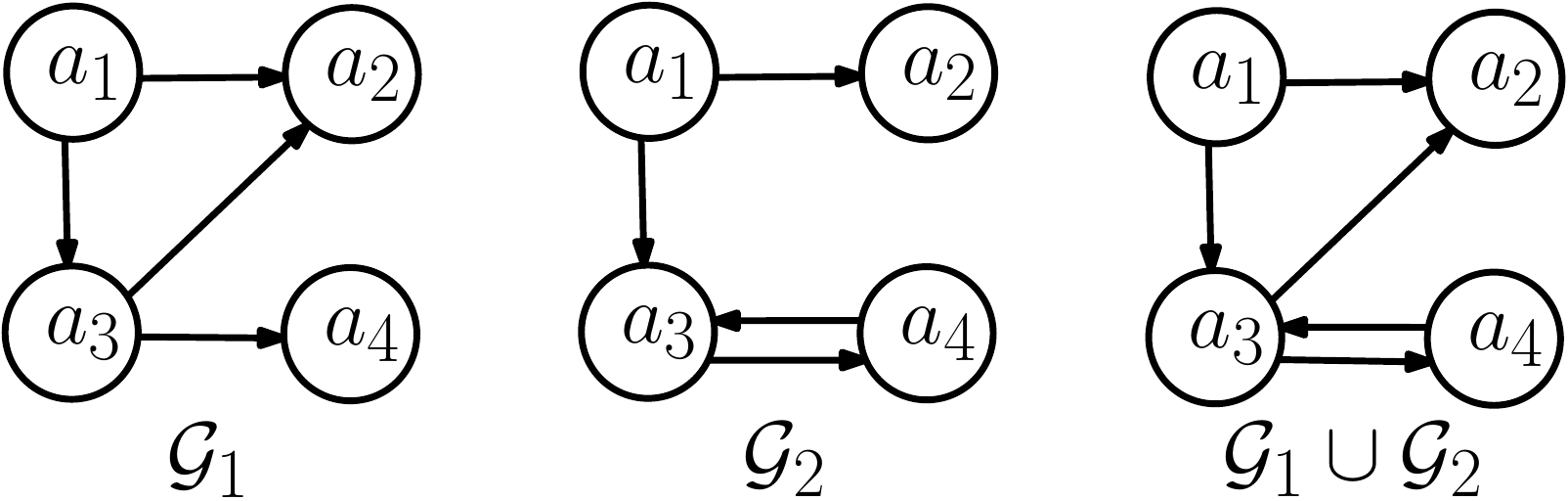}\end{center}

For this example, let $A^*_{goal}=\{a_2\}$.  Almost sure convergence to $a_2$ is not guaranteed for any initial $a(0) \neq a_2$ if the system always stays in either of the modes. Under $P_1$, there is a possibility of converging to $a_4$. Under $P_2$, there is a possibility of converging to the communicating class $\{a_3,a_4\}$ and persistently moving between those two states. Note that there is a feasible path from any state to $a_2$ on the union graph. Hence, in light of Theorem \ref{conv2a}, there exists a switching policy that achieves almost sure convergence to the common absorbing state $a_2$. For example, it can be shown that the following switching policy, which satisfies \eqref{com-pot}, achieves such a convergence guarantee:
\begin{equation}
    \label{polex1}
\sigma(a_1)=1, \; \sigma(a_2)=1, \; \sigma(a_3)=1, \; \sigma(a_4)=2.
\end{equation}
\end{example}

\begin{remark}
Given a set of modes $\pi=\{P_1, \hdots, P_k\}$ and a switching policy ${\sigma: A \mapsto \{1,2,\hdots,k\}}$ that achieves almost sure convergence to  $A^*_{goal}$ (e.g., any policy satisfying \eqref{com-pot}), one important quantity is the expected time to reach $A^*_{goal}$ from any initial state. Once such a switching  policy $\sigma$ is fixed, the system can be treated as a time-invariant absorbing chain with a transition matrix $Q_\sigma$ whose entries are
\begin{equation}
    [Q_\sigma]_{ij}=[P_{\sigma(a_i)}]_{ij}.
\end{equation}
 Accordingly, for any initial state, the expected time to reach $A^*_{goal}$ under $\sigma$ can be computed by applying the methods for time-invariant absorbing chains (e.g., see \cite[Ch. 11.2]{grinstead2012introduction}). 
\end{remark}

\subsection{Stability Under Arbitrary Switching}
In the second part of our analysis, we focus on investigating when the system almost surely converges to an absorbing state from any initial condition under any switching signal. We first show that such a convergence guarantee is possible only if all the modes have the same set of absorbing states. 
 \begin{lem}
\label{conv-nec} Let $\pi=\{P_1, \hdots, P_k\}$ be transition matrices over a finite state space $A$ with the sets of absorbing states $A^*_1, \hdots, A^*_k$. If $a(t)$ almost surely converges to an absorbing state in $A^*_\cup$ from any initial state under any switching signal $\sigma(t)$, then $A^*_1= \hdots= A^*_k$.
\end{lem} 
\begin{proof}  Convergence to some $a^*_i\in A^*_\cup$ under any arbitrary switching signal requires $a^*_i$ to be a common absorbing state, i.e., $a^*_i \in A_\cap^*$, as otherwise by switching to some mode $j$ such that $a^*_i \notin A^*_j$ the system would eventually leave $a^*_i$. Now, suppose that one of the modes, $i$, has an absorbing state that some other mode, $j$, does not have, i.e., $a_i^* \in A_i^*\setminus A^*_j$. In that case, the system never reaches a common absorbing state in $A_\cap^*$ from $a_i^*$ when the system is always kept in mode $i$, resulting in a contradiction. Hence, almost sure convergence to an absorbing state from any initial state under any switching signal is possible only if $A_1^*=\hdots =A_k^*$.
 \end{proof}

In light of Lemma \ref{conv-nec}, in the remainder of this section we will only consider systems where all the modes have the same set of absorbing states. We will provide sufficient conditions for almost sure convergence under any switching among the modes. Note that such a global convergence guarantee is a very strong property and may not hold even when all the modes are absorbing to the same set of states, i.e., always staying in any single mode would ensure almost sure convergence. We provide such an example below.

 \begin{example}
Consider a system with two modes:
\begin{equation}
    \label{modex2}
P_1=\begin{bmatrix} \times & \times & 0 & 0 \\ 0 & \times & \times & 0 \\ 0 & 0 & \times & \times \\0 & 0 & 0 & 1 \end{bmatrix}, \; P_2=\begin{bmatrix} \times & 0 & \times & 0 \\ 0 & \times & 0 & \times \\ 0 & \times & \times & 0 \\0 & 0 & 0 & 1 \end{bmatrix}.
\end{equation}
The simplified transition graphs $\mathcal{G}_1$ and $\mathcal{G}_2$, are as follows:
\vskip2ex
\begin{center} \includegraphics[trim =0mm 0mm 0mm 0mm, clip,scale=0.55]{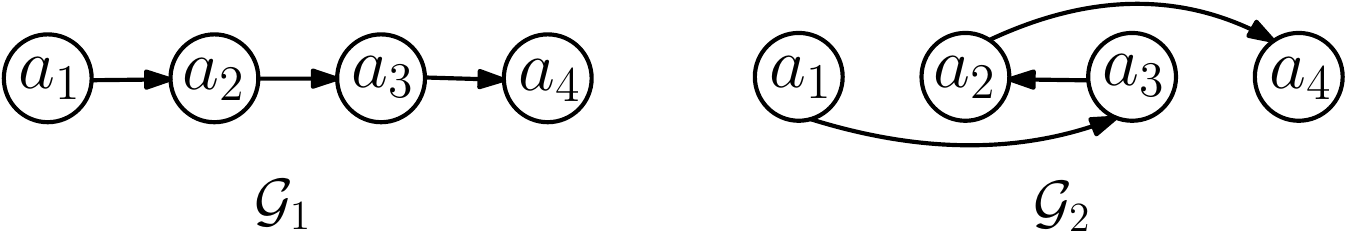}\end{center}

Note that $A^*_1=A^*_2= \{a_4\}$. Furthermore, for any $a(0)\in A$, the state almost surely convergences to $a_4$ if the system never switches and always stays in either of these modes. However, it can be shown that the system would get stuck in a cycle over the states $a_2$ and $a_3$ and never reach $a_4$ from any $a(0) \neq a_4$ under the following switching policy:
\begin{equation}
    \label{polex2}
\sigma(a_1)=1, \; \sigma(a_2)=1, \; \sigma(a_3)=2, \; \sigma(a_4) =2.
\end{equation}
Hence, while both modes are absorbing to the same state, the convergence is not guaranteed under every switching signal.
\end{example}

 \begin{theorem}
\label{conv1} Let $\pi=\{P_1, \hdots, P_k\}$ be transition matrices over a finite state space $A$, each of which has the same set of absorbing states $A^* \neq \emptyset$. Then $a(t)$ almost surely converges to an absorbing state under any switching signal $\sigma(t)$ if any of the following conditions is true:

\begin{enumerate}
    \item The intersection of simplified transition graphs, $\mathcal{G}_\cap$, is weakly acyclic with no sinks outside $A^*$.
    \item The union of simplified transition graphs, $\mathcal{G}_\cup$, is acyclic.
    \item All the simplified transition graphs are weakly acyclic and on each of those graphs every $a \notin A^*$ has a link to some $a'$ with a smaller maximum distance to $A^*$, i.e., 
\begin{equation}
\label{conv31b}
 \forall \mathcal{G}_i =(A,E_i), a\notin A^*, \exists (a,a')\in E_i : \bar{d}(a',A^*)< \bar{d}(a,A^*),
 \end{equation}
where $\bar{d}(a,A^*)$ denotes the maximum distance of $a$ to $A^*$ on the simplified transition graphs, i.e.,
\begin{equation}
\label{conv31}
\bar{d}(a,A^*) = \max_{i\in \{1, \hdots,k\}}d_i(a,A^*).
\end{equation} 
\end{enumerate}
\end{theorem} 
\begin{proof} We prove the sufficiency of each condition.

\underline{Condition 1}: Since all the modes have the same set of absorbing states $A^*$ and $\mathcal{G}_\cap=\mathcal{G}_1 \cap \hdots \cap \mathcal{G}_k$ has no sinks outside $A^*$, the sinks of $\mathcal{G}_\cap$ are the nodes in $A^*$. For any $a \notin A^*$, let $p(a,a^*)$ be the shortest path on $\mathcal{G}_\cap$ from $a$ to some $a^*\in A^*$. Note that this path is feasible under any switching signal since any edge on the intersection graph $\mathcal{G}_\cap$ also exists on all $\mathcal{G}_1, \hdots, \mathcal{G}_k$. Hence, when the system is in state $a$, the probability of path $p(a,a^*)$ being taken by the system (hence reaching $A^*$) is lower bounded by 
\begin{equation}
    \label{lbtr}
    \min_{1 \leq q \leq k, (i,j) \in p(a,a^*) }[P_q]_{ij}^{|p(a,a^*)|},
\end{equation}
which is the $|p(a,a^*)|^{th}$ (number of edges along $p(a,a^*)$) power of the smallest transition probability assigned to an edge in $p(a,a^*)$ in any of the modes. Since the state space is finite, we have $|p(a,a^*)|< \infty$. Hence, the probability of taking path $p(a,a^*)$ is bounded away from zero, irrespective of the switching signal. Note that if the system does not reach $A^*$ within $|p(a,a^*)|$ time steps, then the same argument can be repeated for the new current state. Hence, the probability of never reaching  $A^*$ converges to zero as time increases. Once the system reaches some $a^*\in A^*$, it can never leave since $a^*$ is an absorbing state in all modes. 

\underline{Condition 2}: Since all the modes have the same set of absorbing states $A^*$, the system eventually leaves any state $a \notin A^*$ with probability one, irrespective of the switching signal $\sigma(t)$. This is because any $a \notin A^*$ has at least one outgoing edge (non-zero probability to leave) in every possible mode. Note that any feasible transition of the system between two states $a\neq a' \in A$ is included as an edge on the union of simplified transition graphs, ${\mathcal{G}_\cup =\mathcal{G}_1 \cup \hdots \cup \mathcal{G}_k}$. Since $\mathcal{G}_\cup$ is acyclic, the system can never go back to the same non-absorbing state once it leaves that state. Hence, the state transitions must (with probability one) eventually lead to a sink on $\mathcal{G}_\cup$.
Note that any sink on $\mathcal{G}_\cup$ must have no outgoing edges in any of $\mathcal{G}_1,\hdots,\mathcal{G}_k$. Hence, the sinks of $\mathcal{G}_\cup$ are the absorbing states, $A^*$.

\underline{Condition 3}: If all the simplified transition graphs are weakly acyclic, then each state $a \in A$ has a finite distance to $A^*$ on all those graphs. At time $t$, let the system be at some $a(t) \notin A^*$. If \eqref{conv31b} holds, then no matter what the current mode  $\sigma(t)$ is, there is a non-zero probability that the system transitions into some $a(t+1)$ such that
 \begin{equation}
     \label{conv32}
     \bar{d}(a(t+1),A^*)<\bar{d}(a(t),A^*).
 \end{equation}
 Since the same argument holds for any time $t$, we can always find a finite sequence of such transitions along which $\bar{d}(a(t))$ strictly decreases down to zero. Note that  $\bar{d}(a(t))=0$ if and only if $ a(t) \in A^*$. Accordingly, there is always a non-zero probability that the system will reach an absorbing state within a finite number of time steps. Consequently, the probability of this event never happening converges to zero as time goes to infinity. 
\end{proof}

\begin{remark}
The acyclicity requirement in the second condition of Theorem \ref{conv1} can not be relaxed to weak acyclicity. This can be seen in Example 2,  where the union of the simple transition graphs is weakly acyclic and there exists a switching policy ensuring that the absorbing state will never be reached from the other states.  
\end{remark}

\begin{remark}
The strict inequality requirement in \eqref{conv31b} can not be relaxed as $\bar{d}(a',A^*)\leq \bar{d}(a,A^*)$. This can be seen in Example 2, where the maximum distances of the nodes to the absorbing state are $\bar{d}(a_1,\{a_4\})=3$, $\bar{d}(a_2,\{a_4\})=\bar{d}(a_3,\{a_4\})=2$, $\bar{d}(a_4,\{a_4\})=0$. Both on $\mathcal{G}_1$ and $\mathcal{G}_2$, every $a \notin A^*$ has a link to some $a'$ such that $\bar{d}(a',A^*)\leq \bar{d}(a,A^*)$, and there exists a switching policy that ensures the absorbing state will never be reached from the other states.  
\end{remark}

 Next, we show that each of the three sufficient conditions in Theorem \ref{conv1} has some marginal value for the verification of absorption in time-varying Markov chains under arbitrary switching. In other words, each of these conditions is applicable to some cases that can not be solved by using the other two conditions. We show this by presenting three examples.


\begin{example}
Consider a system with two modes:
\begin{equation}
    \label{modex3}
P_1=\begin{bmatrix} 0 & \times & \times & 0 \\ 0 & \times & \times & 0 \\ \times & 0 & 0 & \times \\0 & 0 & 0 & 1 \end{bmatrix}, \; P_2=\begin{bmatrix} 0 & 1 & 0 & 0 \\ 0 & 0 & \times & \times \\ 0 & 0 & \times & \times \\0 & 0 & 0 & 1 \end{bmatrix}.
\end{equation}
Accordingly, the simplified transition graphs $\mathcal{G}_1$ and $\mathcal{G}_2$, and their intersection and union are as follows:
\vskip2ex
\begin{center} \includegraphics[trim =0mm 0mm 0mm 0mm, clip,scale=0.55]{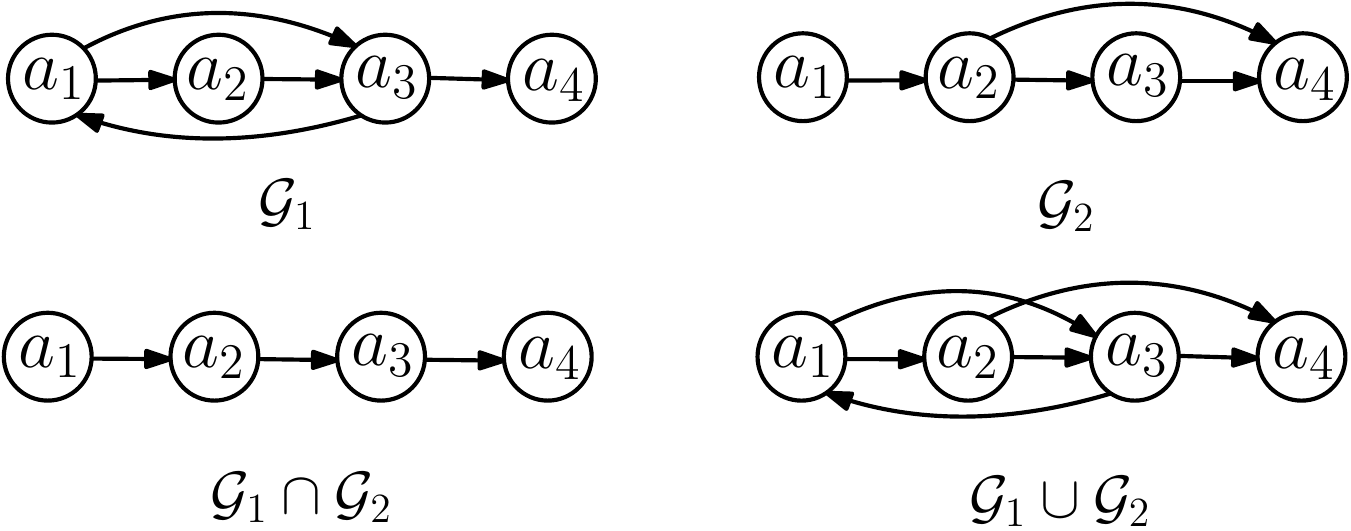}\end{center}

In this case, the second condition in Theorem \ref{conv1} is not applicable since the union graph is not acyclic (it is weakly acyclic). Furthermore, the maximum distances of nodes to the set of absorbing states, $A^*=\{a_4\}$, are as follows: $\bar{d}(a_1,A^*)=\bar{d}(a_2,A^*)=2$, $\bar{d}(a_3,A^*)=1$, $\bar{d}(a_4,A^*)=0$. On $\mathcal{G}_2$, $a_1$ has an outgoing link to only $a_2$, whose maximum distance to $a_4$ is equal to that of $a_1$. Hence, the third condition in Theorem \ref{conv1} is not applicable to this case either. However, 
in light of the first condition in Theorem \ref{conv1}, by inspecting the intersection graph, $\mathcal{G}_1 \cap \mathcal{G}_2$, we know that this system will almost surely converge to $a_4$ from any initial condition ${a(0)\in A}$, irrespective of the switching signal $\sigma(t)$. 

\end{example}

\begin{example}
Consider a system with two modes:
\begin{equation}
    \label{modex4}
P_1=\begin{bmatrix} \times & \times & 0 & 0 \\ 0 & 0 & 0 & 1 \\ 0 & 0 & \times & \times \\0 & 0 & 0 & 1 \end{bmatrix}, \; P_2=\begin{bmatrix} 0 & 0 & 1 & 0 \\ 0 & \times & \times & 0 \\ 0 & 0 & \times & \times \\0 & 0 & 0 & 1 \end{bmatrix}.
\end{equation}
Accordingly, the simplified transition graphs $\mathcal{G}_1$ and $\mathcal{G}_2$, and their intersection and union are as follows:
\vskip2ex
\begin{center} \includegraphics[trim =0mm 0mm 0mm 0mm, clip,scale=0.55]{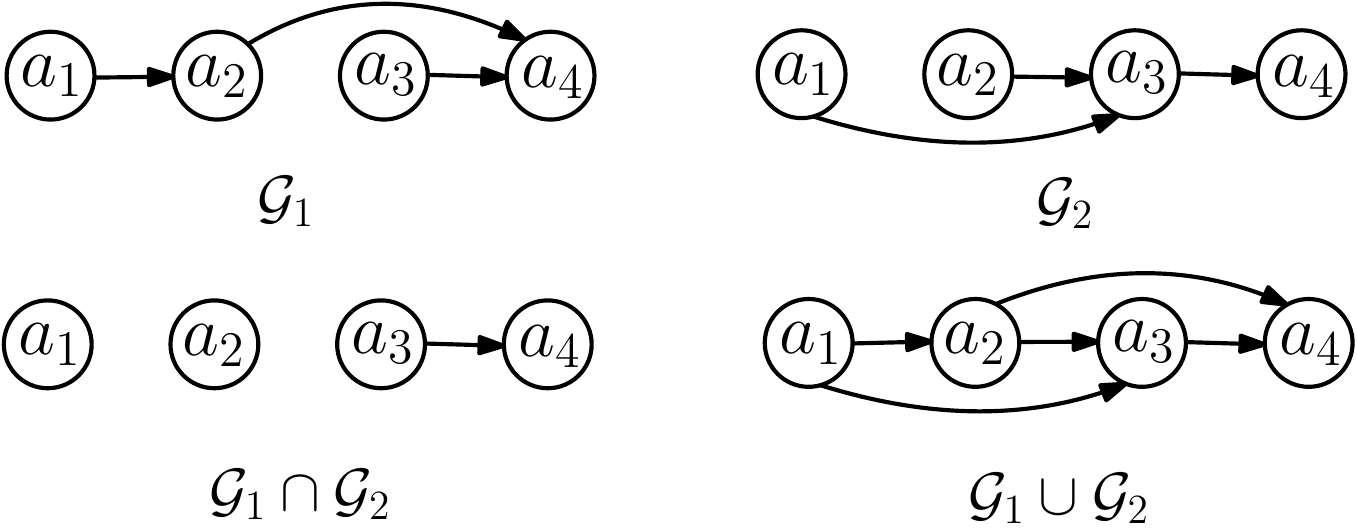}\end{center}

The first condition in Theorem \ref{conv1} is not applicable since ${a_1,a_2 \notin A^*=\{a_4\}}$ are also sinks on the intersection graph. Furthermore, the maximum distances of nodes to the set of absorbing states are as follows: $\bar{d}(a_1,A^*)=\bar{d}(a_2,A^*)=2$, $\bar{d}(a_3,A^*)=1$, $\bar{d}(a_4,A^*)=0$. On $\mathcal{G}_1$, $a_1$ has an outgoing link to only $a_2$, whose maximum distance to $a_4$ is equal to that of $a_1$. Hence, the third condition in Theorem \ref{conv1} is not applicable to this case either. However, 
based on the second condition in Theorem \ref{conv1}, we know that this system will almost surely converge to $a_4$ from any ${a(0)\in A}$, under any $\sigma(t)$ since the union graph, $\mathcal{G}_1 \cup \mathcal{G}_2$, is acyclic. 

\end{example}

\begin{example}
Consider a system with two modes:
\begin{equation}
    \label{modex5}
P_1=\begin{bmatrix} \times & \times & 0 & 0 \\ 0 & \times & \times & 0 \\ 0 & 0 & \times & \times \\0 & 0 & 0 & 1 \end{bmatrix}, \; P_2=\begin{bmatrix} \times & 0 & \times & 0 \\ \times & 0 & 0 & \times \\ 0 & 0 & \times & \times \\0 & 0 & 0 & 1 \end{bmatrix}.
\end{equation}
Accordingly, the simplified transition graphs $\mathcal{G}_1$ and $\mathcal{G}_2$, and their intersection and union are as follows:
\vskip2ex
\begin{center} \includegraphics[trim =0mm 0mm 0mm 0mm, clip,scale=0.55]{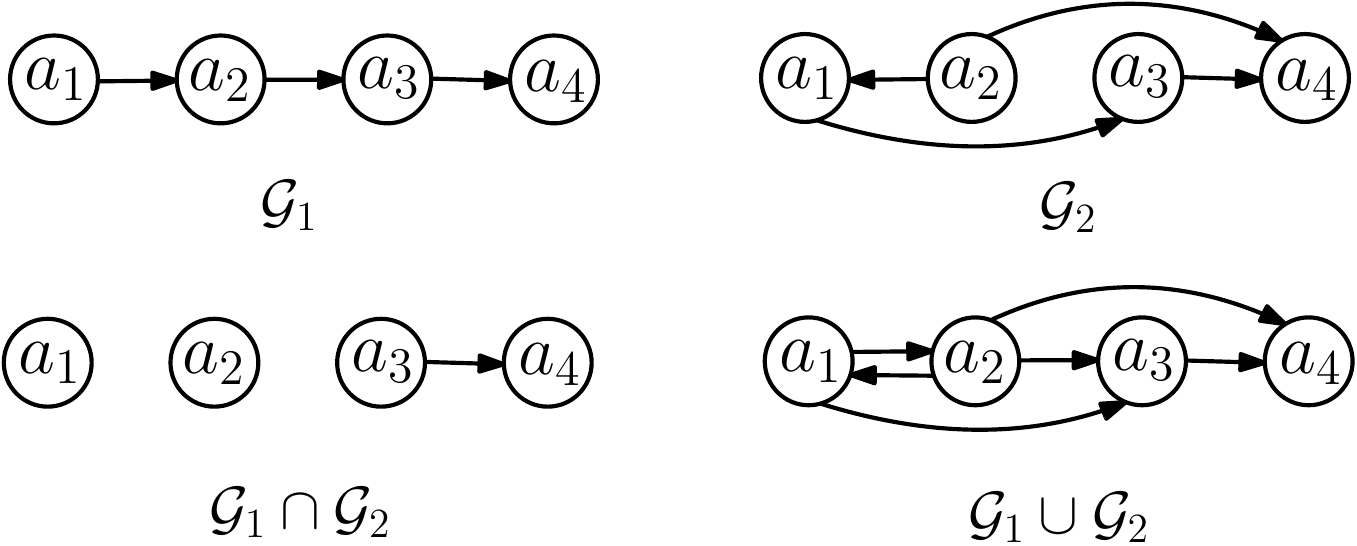}\end{center}

In this example,  the first condition in Theorem \ref{conv1} is not applicable since ${a_1,a_2 \notin A^*=\{a_4\}}$ are also sinks on the intersection graph. Furthermore, the second condition is not applicable since the union graph is not acyclic (it is weakly acyclic). Note that the maximum distances of nodes to the set of absorbing states, $A^*=\{a_4\}$, are as follows: $\bar{d}(a_1,A^*)=3$, $\bar{d}(a_2,A^*)=2$, $\bar{d}(a_3,A^*)=1$, $\bar{d}(a_4,A^*)=0$. Both on $\mathcal{G}_1$ and $\mathcal{G}_2$, each non-absorbing state has an outgoing link to some other state whose maximum distance to $A^*$ is smaller. Accordingly, in light of the third condition in Theorem \ref{conv1}, we know that this system will almost surely converge to $a_4$ from any ${a(0)\in A}$, under any switching signal $\sigma(t)$. 
\end{example}

We conclude this section with a remark regarding the application of our results to Markov chains with infinitely many modes.
\begin{remark}
The results in this paper can also be extended to Markov chains with a finite state space and an infinite set of modes, $\pi$, when the probabilities of feasible transitions are bounded away from zero, i.e., there exists $\epsilon>0$ such that for every mode $P\in \pi$, every non-zero entry ${[P]_{ij}>0}$ satisfies $[P]_{ij}\geq \epsilon$. Note that while the probabilities of feasible transitions are always bounded away from zero for a finite set of modes, this is not necessarily true when there are infinitely many modes. In such cases, the almost sure convergence arguments in the proofs of Theorems \ref{conv2a} and \ref{conv1}, which are based on the existence of feasible finite paths to absorbing states, may no longer be valid. To illustrate this, consider a system with two states $A=\{a_1,a_2\}$ and an infinite set of transition matrices $\{P_0, P_1, \hdots \}$ such that 
$$
[P_0]_{11}=0.5, \; [P_t]_{11}=\frac{\beta_t}{\beta_{t-1}}, \forall t\geq 1,  \; [P_t]_{22}=1, \forall t\geq 0,
$$
where $\beta_t = 0.25 + 0.25^{t+1}$ for all $t\geq 0$.
Accordingly, all the modes have the same simplified transition graph, which has a single edge: $a_1$ to $a_2$. For a system with a finite set of modes, any of the three conditions in Theorem \ref{conv1} would be applicable to such simplified transition graphs for showing almost sure convergence to $a_2$ under any switching signal $\sigma(t)$. However, when the system in this example starts at $a_1$, the probability of never reaching $a_2$ is 0.25 under $\sigma(t)=t$ since $$\prod_{t=0}^n [P_t]_{11}=\beta_n \text{ and } \lim_{n\to \infty}\beta_n=0.25.$$ The lack of almost sure convergence in this example arises from the fact that while ${[P_t]_{12}=1-[P_t]_{11}>0}$ for all $t\geq 0$, it approaches to zero rapidly as $t$ increases.  
\end{remark}

\section{Conclusion}
\label{conc}
We investigated the asymptotic behavior of time-varying (non-homogeneous) discrete-time Markov chains with finite state space. We particularly focused on almost sure convergence to absorbing states in systems that switch among a finite set of transition matrices (modes). We showed that a switching policy that ensures almost sure convergence to a desired set of absorbing states, $A^*_{goal}$, from any initial state exists if and only if $A^*_{goal}$ is reachable from any state on the union of simplified transition graphs. We then showed that almost sure convergence to an absorbing state from any initial condition under any switching is possible only when all the modes have the same set of absorbing states $A^*$. We provided three sufficient conditions for such stability: 1) the intersection of simplified graphs is weakly acyclic and have no sinks other than $A^*$, or 2) the union of simplified transition graphs is acyclic, or 3) in every mode, each state $a_i \notin A^*$ has a feasible transition to some state $a_j \in A$ whose maximum distance (among all simplified transition graphs) to $A^*$ is less than that of $a_i$'s. We also provided examples to show that each of these three sufficient conditions can verify stability in some cases where the other two conditions are not satisfied.


 As a future direction, we plan to explore the applications of our results to the design of provably correct learning, planning, and control algorithms for autonomous systems in stochastic and dynamic environments. One area of interest is game-theoretic learning (e.g., \cite{Yazicioglu13NECSYS,Bhat19}), where standard best-response type algorithms induce a Markov chain over the action space with the Nash equilibria being the absorbing states. We intend to use our results for studying the robustness of stability (convergence to a Nash equilibrium when the utility functions change over time), and the equilibrium selection (convergence to a specific Nash equilibrium by altering the utility functions). We are also interested in applying our results to motion planning under complex specifications represented as autamata-based temporal logics (e.g., \cite{belta2017formal,vasile2017time}) for stochastic systems in dynamic environments, where the goal is to reach an accepting state while the feasible transitions may change over time.
 


\bibliography{MyReferences}

\end{document}